\documentclass[12pt,draftclsnofoot,journal,onecolumn]{IEEEtran}

\usepackage{cite}
\usepackage{graphicx}
\usepackage{amsmath}
\usepackage{amsfonts}
\usepackage{array}
\usepackage{amssymb}
\usepackage{subfigure}
\usepackage[stable]{footmisc}
\usepackage{multirow}
\newtheorem{lemma}{Lemma}
\newtheorem{theorem}{Theorem}
\newtheorem{corollary}{Corollary}

\newtheorem{assumption}{Assumption}

\begin{document}
\title{Downlink Coordinated Multi-Point with Overhead Modeling in Heterogeneous Cellular Networks}
\author{
Ping Xia, Chun-Hung Liu, and Jeffrey G. Andrews\\
\thanks{P. Xia and J. G. Andrews are with the Wireless Networking and Communications Group in the Department of Electrical and Computer Engineering at The
University of Texas at Austin (email: pxia@mail.utexas.edu and jandrews@ece.utexas.edu). C.-H. Liu is with the Department of Engineering Science and Electrical Engineering at National Cheng Kung University, Tainan, Taiwan (email: chungliu@mail.ncku.edu.tw). Dr. Liu is the contact author. This research was supported by Nokia Siemens Networks and the National Science Foundation CIF-1016649. Manuscript last modified: \today}
}
\maketitle

\begin{abstract}
Coordinated multi-point (CoMP) communication is attractive for heterogeneous cellular networks
(HCNs) for interference reduction. However, previous approaches to CoMP face two major hurdles in HCNs. First, they usually ignore the inter-cell overhead messaging delay, although it results in an irreducible performance bound. Second, they consider the grid or Wyner model for base station locations, which is not appropriate for HCN BS locations which are numerous and haphazard. Even
for conventional macrocell networks without overlaid small cells, SINR results are not tractable in the grid model
nor accurate in the Wyner model. To overcome these hurdles, we develop a novel analytical framework which includes the impact of overhead delay for CoMP evaluation in HCNs. This framework can be used for a class of CoMP schemes without user data sharing. As an example, we apply it to downlink CoMP zero-forcing beamforming (ZFBF), and see significant divergence from previous work. For example, we show that CoMP ZFBF does not increase throughput when the overhead channel delay is larger than $60\%$ of the channel coherence time. We also
find that, in most cases, coordinating with only one other cell is nearly optimum for downlink CoMP ZFBF.

\end{abstract}

\section{Introduction}
Improving cellular network capacity is a serious concern of network operators, as new generation mobile devices (e.g. smart phones and mobile connected tablets/laptops) keep getting wider adoption and generating crushing data demand\cite{CISCO12}. One approach among many viable solutions is coordinated multi-point communication, where multiple cells cooperate to improve the key quality-of-service metrics including network throughput (see \cite{Gesbert10CoMPSummary} for an overview). As the cellular network is now in a major transition from conventional macrocell network to heterogeneous cellular network with additional tiers of small cells (e.g. microcells, picocells, femtocells and distributed antennas), it is highly desirable to study CoMP techniques in this new paradigm.

\subsection{CoMP Study Hurdles in Heterogeneous Cellular Networks}
Existing research studies the CoMP concept in conventional macrocell-only networks\cite{Gesbert10CoMPSummary,CadJaf08,SomSha00,ShaZai01,HuaVen04,JinTseSor08}. However, two important aspects of existing works prevent their direct application to the new environment of HCNs.

The first one is the idealized assumption on overhead messaging among neighboring cells, i.e. assuming no inter-cell overhead delay. Not surprisingly, it results in promising predictions on CoMP performance such as multi-fold throughput gain\cite{Gesbert10CoMPSummary,CadJaf08,JinTseSor08,SomSha00,ShaZai01,HuaVen04,Fos06gridmodel}. However, inter-cell overhead delay is not trivial in typical cellular network \cite{Qcomm10CTW,Irmeretal11CoMPTrial,Qcom12ITA,BackhaulImpact,HuaWei09}, which leads to an irreducible performance bound in theory \cite{AliTse10} and significant performance degradations in practice \cite{Qcomm10CTW,Irmeretal11CoMPTrial,Qcom12ITA}. The importance of inter-cell overhead delay is further confirmed by latest industrial implementations, where the performance degradations are much smaller when inter-cell overhead channel is particularly optimized (e.g. directly connecting base stations with gigabyte Ethernet)\cite{Qcom12ITA,Irmeretal11CoMPTrial}. Clearly, inter-cell overhead delay is an important performance limiting factor, and must be modelled and quantified in CoMP study\cite{LozHeaAnd12}. However, this is not a trivial task in HCNs where different types of base stations (BSs) have very different backhaul capabilities and protocols\cite{Qcom12ITA,BackhaulImpact,HuaWei09,And12JSAC}.

The second one is the SINR characterization. Previous works use the grid model or the Wyner model of base station locations to characterize the end-user's signal and other-cell interference in CoMP schemes. Unfortunately, neither model is suitable for HCNs because they assume base stations (BSs) are located on regular positions (e.g. locations form a hexagon\cite{ZhangICIC,SomSha00}, a line\cite{WynerModel} or a circle\cite{JinTseSor08}) while small cells in HCNs have unplanned ad hoc locations. Besides, SINR characterization under these two models are inaccurate or intractable, or both. The Wyner model allows clean-form SINR results, to understand CoMP concept from information-theoretic perspective. However, this model is not accurate due to unrealistic assumptions on the wireless channel and inter-cell interference \cite{Xu11Wyner}. On the other hand, the grid model of conventional macrocell networks is known to be intractable for SINR analysis\cite{Fos06gridmodel}. In HCNs with additional tiers of overlaid small cells, the grid model is  more complex and tractable SINR results are almost hopeless\cite{DhiGanBacAnd11}, not to mention a grid model for small cells is unlikely to be very realistic. In sum, previous models are incapable to capture the new characteristics of BS locations in HCNs, and new models are highly desired for accurate and tractable SINR analysis. 

\subsection{Previous Work}
Early theoretical literature completely ignores the impact of overhead messaging in CoMP schemes, i.e. they assume that overhead messages have no quantization error and the overhead channel is delay-free with infinitely large capacity\cite{SomSha00,Gesbert10CoMPSummary,ShaZai01,HuaVen04,CadJaf08}. Such an ideal assumption is useful for the understanding of CoMP fundamentals, but is obviously far from reality in most practical cases. As a result, it causes highly over-optimistic predictions on the performance of CoMP schemes\cite{Qcomm10CTW,Irmeretal11CoMPTrial,Qcom12ITA}.

Practical issues of overhead messaging are considered in a few more recent works. The capacity limit of inter-cell overhead channel is considered especially in CoMP joint processing where user data is shared among cells\cite{SanSomPooSha09,MarFet07}. The limited feedback model is widely used to characterize the quantization inaccuracy in overhead messages\cite{Loveetal08,JindalRVQ,ZhangICIC}. However, the impact of overhead delay is either ignored or considered under a very simplified model -- that is, a fixed delay model\cite{Ako10Correlated}. In general, appropriated modeling of overhead delay is still missing, to capture the imperfections of overhead channel such as congestion and hardware delays. In our previous work \cite{XiaJoAnd11}, we proposed various models on overhead channels in HCNs (e.g. backhaul and over-the-air overhead channels) and derived the respective delay distributions. These results will be used in this paper to quantify the impact of overhead delay in CoMP performance. 

Because the grid model and the Wyner model are obviously not suitable for SINR characterization in HCNs, several recent works focus on developing new models for BS locations\cite{AndBacGan10,DhiGanBacAnd11}. These works model the locations of BSs in HCNs as nodes in one or more spatial Poisson Point Processes (PPPs). Base station transceiver parameters (e.g. transmit power and path-loss exponent) become the mark of the respective node in the PPP. In this way, the PPP model characterizes the BS location randomness as well as the heterogeneity among different types of BSs. Previous studies on this PPP model show that, besides providing analytical tractability, the new model is at least as accurate as the hexagon-grid model in characterizing the SINR distribution\cite{AndBacGan10,DhiGanBacAnd11,blaszczyszyn2012using,taylor2012pairwise}. Therefore, some in industry have begun to use it for SINR characterization in HCNs\cite{Muk12}. 

\subsection{Contributions}
This paper evaluates downlink CoMP in HCNs, using appropriate models of inter-cell overhead delay and BS locations. We first develop a new analytical framework for the evaluation of a class of CoMP schemes without user data sharing among cells. This framework quantifies the impact of inter-cell overhead delay in HCNs by using our previous results on the delay distributions under various overhead channel configurations\cite{XiaJoAnd11}. Note that this framework includes previous CoMP analysis without overhead delay modeling as a special case. Therefore it can be used to explain the performance gaps between previous analytical predictions and real implementations of CoMP. 

To concretely illustrate the usage of this framework, we apple it to a specific scheme: downlink CoMP zero-forcing beamforming (ZFBF). CoMP ZFBF has been studied in macrocell networks before \cite{JinTseSor08,Dah10beamforming,ZhangICIC} and is attracting industrial implementation efforts\cite{Qcom12ITA}. We derive upper and lower bounds on the end-user SIR distribution for CoMP ZFBF in HCNs, using the spatial PPP model to characterize other-cell interference from all BSs in the entire plane. These bounds are closed-form and show clear dependence on important parameters such as the overhead message bit size and the number of coordinated cells. Using this SIR characterization along with the CoMP evaluation framework, we quantify the downlink CoMP ZFBF coverage and throughput as functions of overhead messaging configurations. Compared with previous work, our results provide new design insights for CoMP ZFBF, for example, on the best number of coordinated cells and the appropriate configuration of overhead channels. 

\section{System Model}
\subsection{Downlink Heterogeneous Cellular Network Modeling}
We consider a downlink heterogeneous cellular network consisting of $K$ different types of base stations (e.g. macrocells, microcells, picocells, femtocells and distributed antennas). We refer to a specific type of BSs as a tier, and thus call the network a $K$-tier HCN. For the $k$-th tier, the BSs have transmit power $P_{k}$, number of antennas $N_{k}$, path loss exponent $\alpha_{k}$ and spatial density $\lambda_k$ BSs per unit area. For example, compared with macrocells, femtocells typically have much lower transmit power, fewer antennas and eventually a much higher spatial density as tens to hundreds of femtocells will often be deployed in the area of a macrocell\cite{And12JSAC}.

We consider a typical end-user equipped with a single antenna. We denote its location as the origin and the locations of BSs as $\{X_{i,k}, k= 1,2\ldots,K, \, i\in \mathbb{N}\}$, where $X_{i,k}$ is the location of the $i^{th}$ closest BS to the origin in the $k^{th}$ tier. We assume all tiers are independently distributed on the plane $\mathbb{R}^2$ and BSs in the $k^{th}$ tier are distributed according to homogeneous Poisson Point Process (PPP) $\Phi_k$ with intensity $\lambda_k$.  

For the purpose of cell association, the end-user will listen to the downlink pilot signals from different BSs, and measure their long-term average powers.
With short-term fading averaged out, the end-user will associate with the BS from whom it receives the strongest average power $\max\limits_{k=1,2\ldots,K, i \in \mathbb{N}} \{P_k |X_{i,k}|^{-\alpha_k} \}$. The index of the selected serving BS is
\begin{align}
\{i^\star,k^\star\} &= \arg \max_{k=1,2,\ldots K,\; i\in \mathbb{N}} \{P_k |X_{i,k}|^{-\alpha_k}\} = \arg \max_{k=1,2,\ldots, K} \{i^\star=1, P_k |X_{1,k}|^{-\alpha_k}\}.
\end{align}
Therefore, we denote the selected serving BS as $\mathrm{BS}_{1,k^\star}$. When $\mathrm{BS}_{1,k^\star}$ is capable of coordinating with $L$ other cells, it usually would prefer to coordinate with the $L$ strongest interfering cells, for example, to maximally cancel inter-cell interference in CoMP zero-forcing beamforming \cite{Gesbert10CoMPSummary,ZhangICIC}. We denote the set of these $L$ coordinated BSs as $\mathcal{B}_L$.

\subsection{Overhead Messaging in CoMP Schemes}
In this paper, we investigate the impact of realistic inter-cell overhead signaling on a class of CoMP schemes, where an end-user is served by only one BS without user data sharing among coordinated cells. To cooperate with the serving cell, a coordinated BS $\mathrm{BS}_{i,k} \in \mathcal{B}_L$ needs to be updated with certain key parameters in that cell. The choices of these cooperation-dependent parameters are different among various CoMP schemes, with common examples being user channel states and user scheduling information. These parameters naturally fluctuate over time because of the dynamics in network environment and user mobility. 
\begin{assumption}\label{AST}
\emph{(i.i.d. discrete time model) We assume the cooperation-dependent parameters stay constant for a time $\mathcal{T}_{i,k}$ (either \emph{deterministic} or \emph{random}) and then change to a new i.i.d. value. We denote $\eta_{i,k}=1/\mathbb{E}[\mathcal{T}_{i,k}]$ as the average change rate.}
\end{assumption}
This assumption is true for parameters such as user scheduling information, which are determined by BSs and stays constant per transmission time interval (TTI). On the other hand, parameters such as channel fading may change continuously. However, the i.i.d. discrete time model on channel fading (named \emph{block fading} model in other literature) is fairly accurate and widely used\cite{TseWC,Loveetal08}.

We assume a parameterized gamma distribution on $\mathcal{T}_{i,k}$ in Assumption \ref{AST}
\begin{align} \label{EquTModel}
\mathcal{T}_{i,k} \sim \mbox{Gamma} \left(M, \frac{1}{M \eta_{i,k}}\right),
\end{align}
where $M$ is the parameter. Note that $\mathbb{E}[\mathcal{T}_{i,k}]=\frac{1}{\eta_{i,k}}$ is unchanged under various values of $M$. This general distribution with different values of $M$ can model different scenarios from deterministic process (i.e. $\mathcal{T}_{i,k}$ becomes deterministic as $M \rightarrow \infty$) to Poisson process (i.e. $\mathcal{T}_{i,k}$ is exponentially distributed for $M=1$).

Because of the dynamics of these cooperation-dependent parameters, the end-user or its serving cell needs to detect the changes in their values. Several parameters like user scheduling information are determined by the BS and their values are thus known on real-time basis. For other parameters such as channel fading, the end-user can constantly measure their values through frequent pilot signals. Once the values of the parameters change, an overhead message will then be generated and sent. In this way, the overhead message will be updated every $\mathcal{T}_{i,k}$, a sufficient frequency without unnecessary burden on the overhead channel. 

After being generated by the serving base station $\mathrm{BS}_{1,k^\star}$, an overhead message is transmitted to $\mathrm{BS}_{i,k}$ through inter-cell overhead channel. The overhead channel incurs delay denoted as $\mathcal{D}_{i,k}$. With this updating overhead message, $\mathrm{BS}_{i,k}$ can take the appropriate cooperation strategy accordingly. Note that each overhead message only has a lifetime of $\mathcal{T}_{i,k}$, because the parameters change after $\mathcal{T}_{i,k}$ and a new overhead message is generated.

We now consider overhead design in CoMP ZFBF, where coordinated neighboring BSs use zero-forcing precoders to null their interference\cite{ZhangICIC,Gesbert10CoMPSummary}. Therefore, the cooperation-dependent parameter in CoMP ZFBF is the channel direction information (CDI)
\begin{align}
\mathbf{\tilde{h}}_{i,k}\stackrel{\triangle}{=}\frac{\mathbf{h}_{i,k}}{\|\mathbf{h}_{i,k}\|},
\end{align}
where $\mathbf{h}_{i,k}$ is the $N_k \times 1$ fading vector between $\mathrm{BS}_{i,k}$ and the end-user. In this paper, we assume uncorrelated Rayleigh fading, i.e. each component of $\mathbf{h}_{i,k}$ is i.i.d. complex Gaussian $\mathcal{CN}(0,1)$. According to Assumption \ref{AST}, the fading $\mathbf{h}_{i,k}$ stays constant for a time $\mathcal{T}_{i,k}$ and then changes independently, i.e. block fading\cite{TseWC}. In this specific example, $\mathcal{T}_{i,k}$ is the channel coherence time. 

In the beginning of each fading block, the end-user observes the new fading value $\mathbf{h}_{i,k}$. It searches through a codebook $\mathcal{C}_{i,k}$ known by both itself and $\mathrm{BS}_{i,k}$, which consists of $2^{B_{i,k}}$ codewords. From the codebook $\mathcal{C}_{i,k}$, the end-user will choose the codeword $\mathbf{c}_{i,k}$ closest to current fading direction $\tilde{\mathbf{h}}_{i,k}$ such that $|\mathbf{\tilde{h}}_{i,k} \mathbf{c}_{i,k}|$ is maximized. The index of this selected codeword $\mathbf{c}_{i,k}$ is fed back to $\mathrm{BS}_{1,k^\star}$ using $B_{i,k}$ bits. For $\mathrm{BS}_{1,k^\star}$, the overhead messages form an arrival process with inter-arrival time $\mathcal{T}_{i,k}$. The serving BS $\mathrm{BS}_{1,k^\star}$ then transmits these overhead messages to $\mathrm{BS}_{i,k}$ through an inter-cell overhead channel. Based on the received overhead, $\mathrm{BS}_{i,k}$ chooses a zero-forcing precoding vector $\mathbf{f}_{i,k}$ such that $|\mathbf{f}_{i,k} \, \mathbf{c}_{i,k}|^2=0$. See Fig. \ref{PicConceptSystem} for a conceptual plot of overhead messaging in CoMP ZFBF.

\begin{figure}[h]
\centerline{
\includegraphics[width=4in]{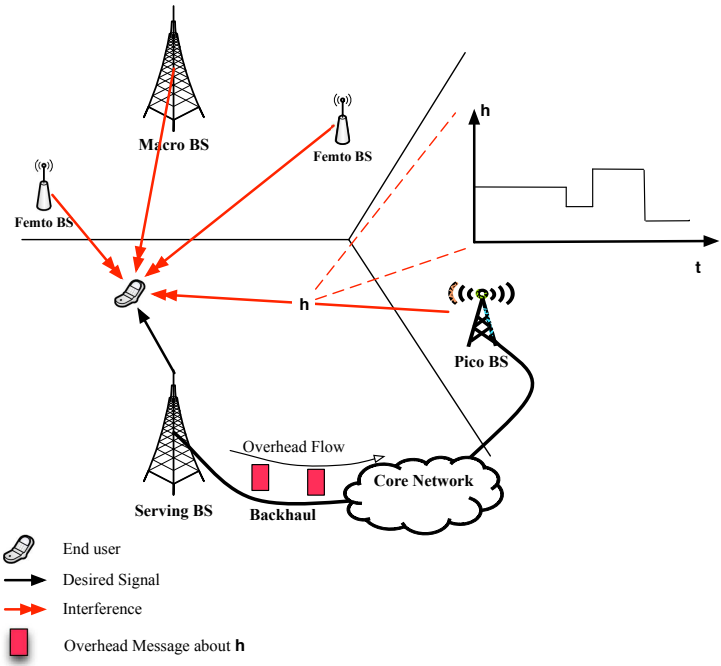}}
\caption{A conceptual plot of CoMP ZFBF in a heterogeneous cellular network. The end-user's serving BS coordinates with a pico BS. As the channel fading $\mathbf{h}$ between the pico BS and the end-user changes over time, the serving BS will perform frequent overhead messaging through backhaul, to inform the pico BS about the current fading values.} \label{PicConceptSystem}
\end{figure}

\subsection{The Impact of Overhead Delay}
Realistic overhead messaging has two major imperfections -- delay and quantization inaccuracy. To make the discussion more concrete, we describe the impact of overhead delay in the context of CoMP ZFBF, while the situations in other CoMP schemes are similar. The delay $\mathcal{D}_{i,k}$ of an overhead message is defined as the time between when that overhead is generated (i.e. the beginning of the respective fading block) and when it is received by $\mathrm{BS}_{i,k}$. It is caused by unavoidable propagation time and the imperfections of the overhead channel such as congestion and hardware delays\cite{XiaJoAnd11}. We call this time window the \textit{overhead messaging phase}. If the overhead delay $\mathcal{D}_{i,k}$ is smaller than the fading block length $\mathcal{T}_{i,k}$, we call the rest time $\mathcal{T}_{i,k}-\mathcal{D}_{i,k}$ in that fading block the \textit{cooperation phase}. Note that the cooperation phase may not exist in a fading block if the overhead delay $\mathcal{D}_{i,k}$ is larger than $\mathcal{T}_{i,k}$. See Fig. \ref{PicCoMPSystem} for an example.

\begin{figure}[h]
\centerline{
\includegraphics[width=4in]{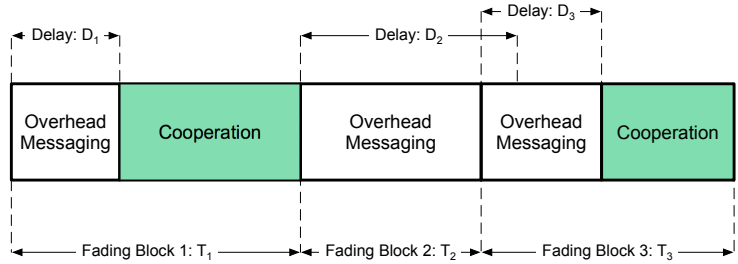}}
\caption{Overhead messaging phases and cooperation phases of a coordinated BS in CoMP ZFBF. The overhead message phase starts from the beginning of each fading block and has a time length of overhead delay $\mathcal{D}$. In a fading block, the coordinated BS will have cooperation phase only if the overhead delay $\mathcal{D}$ is smaller than the fading block length $\mathcal{T}$. } \label{PicCoMPSystem}
\end{figure}

The interference from $\mathrm{BS}_{i,k}$ is different between these two phases. In the overhead messaging phase, the channel fading $\mathbf{h}_{i,k}$ has already changed but $\mathrm{BS}_{i,k}$ does not know its value yet. Therefore the zero-forcing precoder $\mathbf{f}_{i,k}$ is still determined from previously received overhead message. According to Assumption \ref{AST}, the current fading state $\mathbf{h}_{i,k}$ is independent of the previous fading block and thus the precoder $\mathbf{f}_{i,k}$ based on it. Therefore, statistically the interference $|\mathbf{f}_{i,k}\mathbf{h}_{i,k}|^2$ is not reduced in the overhead message phase, which is the worst-case interference scenario\cite{JindalRVQ,Loveetal08}.

On the other hand, $\mathrm{BS}_{i,k}$ receives the new overhead message in the cooperation phase and adjusts its zero-forcing precoder $\mathbf{f}_{i,k}$ accordingly. Because of Assumption \ref{AST}, the fading value $\mathbf{h}_{i,k}$ is assumed to keep unchanged during this block. Therefore the new overhead message remains accurate (minus quantization errors) in the cooperation phase and the selected precoder $\mathbf{f}_{i,k}$ minimizes the interference $|\mathbf{f}_{i,k}\mathbf{h}_{i,k}|^2$ for the entire phase\cite{JindalRVQ,Loveetal08}. This is the best-case interference scenario.

It is seen that Assumption \ref{AST} simplifies the impact of overhead delay into two opposite extreme cases: the interference to $\mathrm{BS}_{i,k}$ is either not reduced in the overhead messaging phase or maximally reduced in the cooperation phase. In practice, the channel fading $\mathbf{h}_{i,k}$ continuously changes with temporal correlation, instead of the i.i.d. discrete-time model in Assumption \ref{AST}. Therefore, the interference from $\mathrm{BS}_{i,k}$ should gradually change over time and is in fact bounded by the two extreme cases. However, it is hard to quantify mathematically. We use Assumption \ref{AST} because it is tractable for analysis, widely used in previous literature on MIMO systems \cite{Loveetal08,TseWC} and allows a first-order analysis on the impact of overhead delay. Future work should consider more complicated models on channel fading, such as other discrete time models in \cite{Wang95MarkovModel,Ako10Correlated,Isu07Correlated}.

\subsection{The Impact of Overhead Quantization Error}
Another concern of realistic overhead messaging is the finite overhead quantization bits $B_{i,k}$. Larger $\mathrm{B}_{i,k}$ usually translates into smaller quantization error and thus higher cooperation gains. However, the exact impact of $\mathrm{B}_{i,k}$ depends on the specific CoMP scheme and the overhead codebook $\mathcal{C}_{i,k}$. See \cite{Loveetal08} for an overview. We now discuss its impact in the context of CoMP ZFBF.

The end-user's SIR $\gamma$ in CoMP ZFBF can be generally expressed as\footnote{In modern cellular networks, thermal noise is not an important consideration either in cell interior where the signal power is strong or in cell edge where interference is usually much larger. Interference is more dominant especially in HCNs because of additional overlaid cells with high spatial densities. We therefore neglect thermal noise and consider SIR in this paper.}
\begin{align} \label{EquPrecoderSIR}
\gamma = \frac{P_{k^\star} |X_{1,k^\star}|^{-\alpha_{k^\star}} |\mathbf{f}_{1,k^\star} \mathbf{h}_{1,k^\star}|^2}{ \sum_{
X_{i,k} \in \bigcup\limits_{k=1}^K \Phi_k \setminus \{ X_{1,k^\star}\}} P_k |X_{i,k}|^{-\alpha_k} |\mathbf{f}_{i,k} \mathbf{h}_{i,k}|^2},
\end{align}
where the value of $|\mathbf{f}_{i,k}\mathbf{h}_{i,k}|$ is elaborated on in the following.

\begin{enumerate}
\item The end-user's serving BS $\mathrm{BS}_{1,k^\star}$ needs to null its interference to the $L$ coordinated cells. Meanwhile it also wants to maximize the signal power $|\mathbf{f}_{1,k^\star} \mathbf{h}_{1,k^\star}|^2$ to the end-user. Thus its precoder $\mathbf{f}_{1,k^\star}$ is chosen such that $|\mathbf{f}_{1,k^\star} \mathbf{h}_{1,k^\star}|^2 \sim \chi_{2N_{k^\star}-2L}^2$\cite{ZhangICIC}. 
\item A coordinated BS $\mathrm{BS}_{i,k} \in \mathcal{B}_L$ cannot null its interference during the overhead messaging phase, because its zero-forcing precoder $\mathbf{f}_{i,k}$ is independent of $\mathbf{\tilde{h}}_{i,k}$ and we have $|\mathbf{f}_{i,k} \mathbf{h}_{i,k}|^2 \sim  \chi_2^2$. In the cooperation phase, $\mathrm{BS}_{i,k}$ receives the updated overhead indicating the codeword $\mathbf{c}_{i,k} = \arg \max_{\mathbf{c} \in \mathcal{C}_{i,k}}  |\mathbf{\tilde{h}}_{i,k} \mathbf{c}|$ and chooses $\mathbf{f}_{i,k}$ such that $|\mathbf{f}_{i,k} \, \mathbf{c}_{i,k}|^2=0$. However, because $\mathbf{c}_{i,k}$ is not the exact CDI $\mathbf{\tilde{h}}_{i,k}$, the value of $|\mathbf{f}_{i,k} \mathbf{h}_{i,k}|^2$ is still positive and dependent on the design of overhead codebook $\mathcal{C}_{i,k}$. In this paper, we assume random vector quantization (RVQ) codebook $\mathcal{C}_{i,k}$, which is commonly used in previous CoMP ZFBF\cite{JindalRVQ,ZhangICIC,Loveetal08}. Under this assumption, we then have $|\mathbf{f}_{i,k} \mathbf{h}_{i,k}|^2 \sim 2^{-\frac{B_{i,k}}{N_k-1}}$ for $\mathrm{BS}_{i,k}$ in the cooperation phase.
\item A non-coordinated BS $\mathrm{BS}_{i,k} \notin \mathcal{B}_L$ chooses the precoder independent of its interference to the end-user, and will have $|\mathbf{f}_{i,k} \, \mathbf{h}_{i,k}|^2 \sim \chi_2^2$.
\end{enumerate}

Based on the derivations above, the end-user SIR can be expressed as 
\begin{align} \label{EquGeneralSIR}
    \gamma=\frac{P_{k^\star} S_{1,k^\star} |X_{1,k^\star}|^{-\alpha_{k^\star}} }{ \sum_{X_{i,k}\in \bigcup\limits_{k=1}^K \Phi_k \setminus \{X_{1,k^\star}\}}P_k \rho_{i,k} S_{i,k}|X_{i,k}|^{-\alpha_k}  }
\end{align} 
where $S_{1,k^\star} \sim \chi^2 (2N_{k^\star}-2L)$, $S_{i,k} \sim  \chi_2^2$, and $\rho_{i,k}$ is the \emph{interference cancellation factor} for $\mathrm{BS}_{i,k}$ 
\begin{align} \label{EquRho}
\rho_{i,k} = 
\begin{cases}
1 & \mbox{For $\mathrm{BS}_{i,k} \notin \mathcal{B}_L$}\\
1 &  \mbox{For $\mathrm{BS}_{i,k} \in \mathcal{B}_L$ during the overhead messaging phase} \\
2^{-\frac{B_{i,k}}{N_k-1}} & \mbox{For $\mathrm{BS}_{i,k} \in \mathcal{B}_L$ during the cooperation phase}\\ 
\end{cases}
\end{align}

\section{CoMP Throughput Evaluation With Imperfect Overhead Messaging}
In this section, we present a throughput evaluation framework for a class of downlink CoMP schemes, where inter-cell overhead messaging does not include user data. Common examples in this category include CoMP beamforming\cite{ZhangICIC,Dah10beamforming} and interference alignment\cite{CadJaf08}. The general framework considers the practical issues from overhead messaging but is also analytically tractable. 

\begin{lemma} \label{LemmaTimeFraction}
The long-term time fraction $\tau_{i,k}$ that a coordinated BS $\mathrm{BS}_{i,k}$ is in the cooperation phase is
\begin{align}\label{EquLemma}
\tau_{i,k}=p(\mathcal{T}_{i,k},\infty)-\eta_{i,k} \int\limits_{0}^{\infty} \left\{p(\mathcal{T}_{i,k},\infty)-p(\mathcal{T}_{i,k},s)\right\} \mathrm{d}s,
\end{align}
where $p(\mathcal{T}_{i,k},s)\triangleq\mathbb{P} (\mathcal{D}_{i,k} \leq \mathcal{T}_{i,k}, \mathcal{D}_{i,k} \leq d)$ is overhead delay distribution.
\end{lemma}
\begin{IEEEproof}
See Appendix \ref{ProofLemmaTimeFraction}.
\end{IEEEproof} 

In our previous work\cite{XiaJoAnd11}, we provide general models on overhead messaging delay for both backhaul and over-the-air inter-cell overhead channels. We then derive the delay distribution $p(\cdot)$ as an explicit function of inter-cell overhead channel parameters (e.g. over-the-air overhead channel bandwidth). Note that our overhead delay models only affect the characterization of $p(\cdot)$, while the result in Lemma \ref{LemmaTimeFraction} holds under various forms of delay distribution $p(\cdot)$.

In previous literature, no overhead delay is considered, i.e. delay $\mathcal{D}_{i,k} = 0$ and thus $p(\mathcal{T}_{i,k},s)=1$ for any $s$ and $\mathcal{T}_{i,k}$. Under this condition, $\tau_{i,k}$ in (\ref{EquLemma}) becomes $1$ as well, which means $\mathrm{BS}_{i,k}$ is always in the cooperation phase with interference maximally reduced. This is obviously an over-optimistic prediction on the interference.

\begin{theorem} \label{ThmFramework}
For CoMP schemes without user data sharing, the end-user's long-term throughput is 
\begin{align} \label{EquMetric}
\mathcal{R} = \sum\limits_{\mathcal{B}\subset \mathcal{B}_L} p_\mathcal{B} R (\gamma_{\mathcal{B}}),
\end{align}
where the summation is over all possible subset $\mathcal{B} \subset \mathcal{B}_L$, $\gamma_{\mathcal{B}}$ is the end-user SIR when BSs $\in \mathcal{B}$ are in the cooperation phase and BSs $\in \mathcal{B}_L \setminus \mathcal{B}$ are in the overhead messaging phase, $R(\cdot)$ is the SIR-rate mapping function, and 
\begin{align}
 p_\mathcal{B}= \left(\prod\limits_{\mathrm{BS}_{i,k} \in \mathcal{B}} \tau_{i,k} \prod\limits_{\mathrm{BS}_{i,k} \in \mathcal{B}_L \setminus \mathcal{B}} (1-\tau_{i,k})\right),
\end{align}
\end{theorem}
\begin{IEEEproof}
Due to the overhead messaging delay, each coordinated BS $\mathrm{BS}_{i,k} \in \mathcal{B}_L$ now has two states: 1) the overhead messaging phase (with probability $1-\tau_{i,k}$) when the overhead message at $\mathrm{BS}_{i,k}$ is already outdated and the updated overhead has not been received yet; and 2) the cooperation phase (with probability $\tau_{i,k}$) when $\mathrm{BS}_{i,k}$ receives the updated overhead message. $\mathrm{BS}_{i,k}$ has different cooperation performance between these two states, for example, as shown in (\ref{EquRho}) for CoMP ZFBF. Therefore, we use a subset $\mathcal{B} \subset \mathcal{B}_L$ to denote the scenario that only BSs $\in \mathcal{B}$ are in the cooperation phase. The probability of this scenario is 
\begin{align}
p_{\mathcal{B}} = \left(\prod\limits_{\mathrm{BS}_{i,k} \in \mathcal{B}} \tau_{i,k} \prod\limits_{\mathrm{BS}_{i,k} \in \mathcal{B}_L \setminus \mathcal{B}} (1-\tau_{i,k})\right).
\end{align}
Each subset $\mathcal{B} \subset \mathcal{B}_L$ corresponds to a possible scenario regarding which coordinated BSs are in the cooperation phase. The end-user's long-term is the average rate over all possible scenarios.
\end{IEEEproof}

As shown in (\ref{EquMetric}), the CoMP throughput evaluation framework explicitly quantifies the impact of overhead delay through the time fraction $\tau_{i,k}$ and probability distribution $p_\mathcal{B}$. It also includes the impact of finite overhead bit size because the end-user SIR $\gamma_{\mathcal{B}}$ is affected by overhead quantization error, for example, as shown in (\ref{EquRho}) for CoMP ZFBF. In this way, the framework considers the imperfections in the overhead messaging. Combined with the SIR characterizations, it can be used to quantify the throughput and coverage of different CoMP schemes. 

As we mentioned before, previous work does not fully consider practical issues in inter-cell overhead messaging. In particular, they ignore the overhead delay, which is in fact non-trivial in most network environments\cite{Qcom12ITA,HuaWei09,Qcomm10CTW}. To show the importance of modeling and analysing overhead delay in CoMP evaluation, we compare the result in Theorem \ref{ThmFramework} with previous work. Without overhead delay, the coordinated BSs always have the updated overhead and stay in the cooperation phase. Obviously, this is a idealized special case of Theorem \ref{ThmFramework}.

\begin{corollary} \label{CoroIdeal}
Under the assumption of delay-free overhead messaging, the long-term CoMP throughput is 
\begin{align} \label{EquCoroIdeal}
\mathcal{R}=R (\gamma_{\mathcal{B}_L}),
\end{align}
where $\gamma_{\mathcal{B}_L}$ is defined in Theorem \ref{ThmFramework}. 
\end{corollary}
\begin{IEEEproof}
When overhead messaging has no delay, all coordinated BSs will always be in the cooperation phase, i.e. $\tau_{i,k} =1$ for each $\mathrm{BS}_{i,k}\in \mathcal{B}_{L}$. Therefore we have $p_{\mathcal{B}}=1$ if $\mathcal{B}=\mathcal{B}_L$ and $p_\mathcal{B}=0$ otherwise. The end-user's rate is then $\mathcal{R} = R (\gamma_{\mathcal{B}_L})$. 
\end{IEEEproof}

As shown in Corollary \ref{CoroIdeal}, assuming no overhead delay significantly simplifies the analysis of CoMP schemes. However, as we elaborated before and will be shown in the numerical simulations, this assumption is far from the reality and causes the wide gaps between analytical predictions and realistic implementations.

\section{CoMP Zero-Forcing Beamforming Throughput Analysis}
In this section, we derives the distribution of the end-user's SIR $\gamma_{\mathcal{B}}$ for CoMP ZFBF, which will be used with the evaluation framework in Theorem \ref{ThmFramework} to quantify the CoMP ZFBF throughput. We first derive the SIR CDF in 1-tier cellular networks in Theorem \ref{ThmSIR1tier}, and then extend it to the general HCN scenario in Theorem \ref{ThmSIRKtier}. 

For now, we simplify the notation for $1$-tier cellular networks by dropping the tier index $k$. Specifically in $1$-tier networks, the BSs have the same transmit power $P$, number of antennas $N$ and path loss exponent $\alpha$. Their locations $\{X_i, i \in \mathbb{N}\}$ form a PPP $\Phi$ with intensity $\lambda$. Therefore the serving BS is simply $\mathrm{BS}_{1}$ -- that is, the nearest BS to the end user -- and the $L$ coordinated BSs are $\mathcal{B}_L=\{\mathrm{BS}_{2}, \ldots, \mathrm{BS}_{L+1}\}$ -- that is, the second to $(L+1)$ nearest BSs to the end-user. Similarly, we now use notation $\rho_{i}$ and $B_{i}$ as simplified versions of $\rho_{i,k}$ and $B_{i,k}$ in (\ref{EquRho}).

In $1$-tier networks, the end-user SIR $\gamma_\mathcal{B}$ in (\ref{EquMetric}) can be simplified as
\begin{align} \label{Equ1tierSIR}
    \gamma_\mathcal{B} &=\frac{P S_1 |X_1|^{-\alpha}}{\sum\limits_{X_i \in \Phi \backslash\{X_1\}} P \rho_i S_i |X_i|^{-\alpha} }= \frac{S_1 |X_1|^{-\alpha}}{\sum\limits_{X_i \in \Phi \backslash\{X_1\}}  \rho_i S_i |X_i|^{-\alpha}}.
\end{align}
where $\rho_i = 2^{\frac{B_i}{N-1}}$ for $\mathrm{BS}_{i} \in \mathcal{B}$ and $\rho_i=1$ otherwise.

\begin{theorem} \label{ThmSIR1tier}
In $1$-tier cellular networks, the CDF of the end-user SIR $\gamma_{\mathcal{B}}$ in (\ref{EquMetric}) is bounded as
\begin{align} \label{EquThmSIR1tier}
\mathbb{P} (\gamma_{\mathcal{B}} \leq \beta)
\begin{cases}
\leq \frac{\beta}{N-L-1} \sum\limits_{i=2}^{\infty} \rho_{i} \frac{\Gamma\left(1+\frac{\alpha}{2}\right)(i-1)!}{\Gamma\left(i+\frac{\alpha}{2}\right)} & \\
\geq 1-\exp\left\{-\Gamma(1+2/\alpha) \left[\frac{(N-L)\Gamma(1-\alpha/2)}{\beta \left[3^{-\alpha}\rho_{\min} + (2l+3)^{-\alpha} (1-\rho_{\min})  \right]}\right]^{-2/\alpha} \right\} & \\
\end{cases}
\end{align}
where $l=|\mathcal{B}|$ is the cardinality of the set $\mathcal{B}$, $\rho_{\min} = \min\limits_{\mathrm{BS}_i \in \mathcal{B}} \{\rho_i\}$, and  
\begin{align}
\rho_i &= 
\begin{cases}
1 & \mathrm{BS}_i \notin \mathcal{B} \\
2^{-\frac{B_i}{N-1}} & \mathrm{BS}_i \in \mathcal{B}\\
\end{cases}
\end{align}

\end{theorem}
\begin{proof}
We denote the normalized interference (normalized by $P$) as
\begin{align} \label{EquInterf1tier}
I_{\mathcal{B}}&\stackrel{\triangle}{=} \sum\limits_{X_i \in \Phi \setminus \{X_1\}} \rho_i S_i |X_i|^{-\alpha} =\sum\limits_{\mathrm{BS}_i \in \mathcal{B}} \rho_i S_i |X_i|^{-\alpha} + \sum_{\mathrm{BS}_i \in \Phi \backslash\{\mathrm{BS}_1 \bigcup\mathcal{B}\}} S_i |X_i|^{-\alpha}.
\end{align}
The last equality comes from the definition of $\mathcal{B}$.  

\textbf{Upper Bound:} For a PPP $\Phi=\{X_1,X_2,
\ldots\}$ in $\mathbb{R}^2$ with an arbitrary density $\lambda$, $\{|X_1|^2, |X_2|^2, \ldots\}$ form a one dimensional PPP with intensity $\pi \lambda$. We thus have 
\begin{align} \label{EquForUB1tier}
\mathbb{E}\left[\frac{|X_1|^\alpha}{|X_i|^\alpha}\right] = \mathbb{E}\left[\left(\frac{|X_1|^2}{|X_i|^2}\right)^{\alpha/2}\right]=\frac{\Gamma\left(1+\frac{\alpha}{2}\right)(i+1)!}{\Gamma\left(i+\frac{\alpha}{2}\right)},
\end{align}
where the last equality holds from Appendix \ref{ProofThmSIR1tierUB}. The upper bound on the SIR CDF is
\begin{align} \label{EquUB1tier}
\mathbb{P} \left(\gamma = \frac{S_1 |X_1|^{-\alpha}}{I_{\mathcal{B}}}\leq \beta\right) 
=\mathbb{P} \left(\frac{I_{\mathcal{B}}}{S_1 |X_1|^{-\alpha}} \geq \frac{1}{\beta}\right) \stackrel{(a)}{\leq} \beta \mathbb{E} \left[\frac{1}{S_1} I_{\mathcal{B}} |X_1|^{\alpha}\right] 
= \beta \mathbb{E}\left[\frac{1}{S_1}\right] \sum\limits_{i=2}^{\infty} \rho_i \mathbb{E}[S_i] \mathbb{E}\left[\frac{|X_1|^{\alpha}}{|X_i|^{\alpha}}\right]
\end{align}
where (a) follows from Markov's inequality. As elaborated below (\ref{EquGeneralSIR}), we have $\mathbb{E}\left[\frac{1}{S_1}\right]=\frac{1}{N-L-1}$ (because $S_1\sim \chi^2_{2N-2L}$) and $\mathbb{E}[S_i]=1$ (because $S_i\sim \chi^2_2$) for $i\geq2$. Therefore  Combing (\ref{EquForUB1tier}) and (\ref{EquUB1tier}) gives the CDF upper bound in Theorem \ref{ThmSIR1tier}.

\textbf{Lower Bound:} Let $I_{(m)}= \sum_{X_i \in \Phi \backslash\{X_1,\ldots,X_m\}} S_i |X_i|^{-\alpha}$. In other words, $I_{(m)}$ is the sum interference experienced by the end-user, if the nearest $m$ BSs (including the serving BS) are removed or turned off. Note that $I_{(0)}$ means that the end-user is not associated with any BS, and experiences interference from all the BSs in the PPP $\Phi$. Denote $l=|\mathcal{B}|$ as the cardinality of $\mathcal{B}$ and $\rho_{\min}=\min_{\mathrm{BS}_i \in \mathcal{B}} \{\rho_i\}$. The actual interference $I_\mathcal{B}$ can be expressed as a function of $I_{(m)}$
\begin{align} \label{EquItoIm}
I_\mathcal{B}
&\geq \sum\limits_{\mathrm{BS}_i \in \mathcal{B}} \rho_{\min} S_i |X_i|^{-\alpha} + \sum_{X_i \in \Phi \backslash\{X_1 \bigcup \mathcal{B}\}} S_i |X_i|^{-\alpha} \nonumber \\
& \stackrel{(b)}{\succeq} \sum\limits_{\mathrm{BS}_i \in \{\mathrm{BS}_{2},\ldots,\mathrm{BS}_{l+1}\}} \rho_{\min} S_i |X_i|^{-\alpha} + \sum_{X_i \in \Phi \backslash\{X_1, \ldots, X_{l+1} \}} S_i |X_i|^{-\alpha} \nonumber \\
& =\rho_{\min} \sum\limits_{X_i \in \Phi \setminus \{X_1\}} S_i |X_i|^{-\alpha} + (1-\rho_{\min}) \sum\limits_{X_i \in \Phi \setminus \{X_1,\ldots,X_{l+1}\}} S_i |X_i|^{-\alpha} \nonumber \\
&=\rho_{\min} I_{(1)} + (1-\rho_{\min}) I_{(l+1)} \nonumber \\
&\stackrel{(c)}{\succeq} \left[3^{-\alpha}\rho_{\min} + (2l+3)^{-\alpha} (1-\rho_{\min}) \right] I_{(0)}
\end{align}
where $\succeq$ means stochastic dominance. In the right hand side of (b), we assume the strongest $l$ interfering BSs are cancelling their interference, instead of the $l$ BSs in the set of $\mathcal{B}$. Therefore it is a lower bound on the interference $I_\mathcal{B}$. (c) holds from Appendix \ref{ProofThmSIR1tierLB} deriving the lower bound of an arbitrary $I_{(m)}$.

Based on the lower bound on interference $I$, the SIR CDF can be lower bounded as follows
\begin{align} \label{EquLB1tier}
\mathbb{P} \left(\gamma=\frac{S_1 |X_1|^{-\alpha}}{I_\mathcal{B}}\leq \beta\right)
\geq & \mathbb{P} \left(\frac{S_1 |X_1|^{-\alpha}}{\left[3^{-\alpha}\rho_{\min} + (2l+3)^{-\alpha} (1-\rho_{\min}) \right] I_{(0)}} \leq \beta\right)
\stackrel{(d)}{\geq} & 1-\mathbb{E}\left[e^{-A(Z)}\right]
\end{align}
where $Z\stackrel{\triangle}{=} \frac{S_1 |X_1|^{-\alpha}}{\beta \left[3^{-\alpha}\rho_{\min} + (2l+3)^{-\alpha} (1-\rho_{\min})  \right]}$, $A(Z)=\pi \lambda \Gamma(1+\frac{2}{\alpha})Z^{-\frac{2}{\alpha}}$, and (d) holds from Theorem 1 in \cite{LiuAnd11}. The CDF lower bound is then derived as
\begin{align} \label{EquLBsequel1tier}
\mathbb{P} (\gamma_\mathcal{B} \leq \beta) & \geq 1- \mathbb{E}\left[e^{-A(Z)}\right] \nonumber \\
&\stackrel{(e)}{\geq} 1-\exp\left\{-\pi \lambda \Gamma(1+2/\alpha)(\mathbb{E}[Z])^{-\frac{2}{\alpha}}\right\} \nonumber\\
&=1-\exp\left\{-\Gamma(1+2/\alpha) \left[\frac{(N-L)\Gamma(1-\alpha/2)}{\beta \left[3^{-\alpha}\rho_{\min} + (2l+3)^{-\alpha} (1-\rho_{\min})  \right]}\right]^{-2/\alpha} \right\},
\end{align}
where (e) holds from Jensen's inequality.

\end{proof}

Based on the results in Theorem \ref{ThmSIR1tier} for the $1$-tier networks, we now characterize the end-user's SIR $\gamma_{\mathcal{B}}$ in a general $K$-tier HCN. 
\begin{theorem}\label{ThmSIRKtier}
In a general $K$-tier heterogeneous cellular network, the CDF of the end-user SIR $\gamma_{\mathcal{B}}$ in (\ref{EquMetric}) is bounded as
\begin{equation}\label{EquThmSIRKtier}
\mathbb{P}\left[\gamma_{\mathcal{B}} \leq \beta\right]
\begin{cases} 
\leq 
\frac{\beta}{N_{k^\star}-L-1} \left[\sum\limits_{i=2}^{\infty}  \rho_{i,k^\star} \frac{\Gamma\left(1+ \frac{\alpha_{k^\star}}{2} \right)(i-1)!}{\Gamma\left(i+\frac{\alpha_{k^\star}}{2}\right)}+ \sum\limits_{\begin{subarray}{c}
k=1\\
k \neq k^\star
\end{subarray}}^K \frac{P_k}{P_{k^\star}} \sum\limits_{i=1}^{\infty} \rho_{i,k} \frac{(\lambda_{k^\star} \pi)^{-\frac{\alpha_{k^\star}}{2}}\Gamma\left(1+ \frac{\alpha_{k^\star}}{2} \right)(i-1)!}{(\lambda_k \pi)^{\frac{-\alpha_k}{2}}\Gamma\left(i+\frac{\alpha_k}{2}\right)}\right]\\
\geq 1-\exp\left[-\left(\pi\hat{\lambda}\right)^{1-\frac{\alpha_{k^\star}}{\alpha_{\max}}}\Gamma\left(1+\frac{2}{\alpha_{\max}}\right)\left(\frac{(N_{k^\star}-L) \Gamma \left(1- \frac{\alpha_{k^\star}}{2}\right)}{\beta[3^{-\alpha_{\max}}\rho_{\min}+(2
l+3)^{-\alpha_{\max}}(1-\rho_{\min})]}\right)^{-\frac{2}{\alpha_{\max}}}\right]\\
\end{cases}
\end{equation}
where $l=|\mathcal{B}|$, $\hat{\lambda}= \sum\limits_{k=1}^{K} \lambda_k \left(P_k/P_{k^\star}\right)^{\frac{2}{\alpha_k}}$, $\alpha_{\max}\triangleq \max\{\alpha_1,\ldots,\alpha_K\}$, $\rho_{\min}=\min\limits_{\mathrm{BS}_{i,k} \in \mathcal{B}} \{\rho_{i,k}\}$, and 
\begin{align}
\rho_{i,k} = 
\begin{cases}
1 & \mathrm{BS}_{i,k} \notin \mathcal{B} \\
2^{-\frac{B_{i,k}}{N_k-1}} & \mathrm{BS}_{i,k} \in \mathcal{B} \\
\end{cases}
\end{align}
\end{theorem}

\begin{proof}
In the $K$-tier HCN, the normalized interference $I_{\mathcal{B}}$ can be written as
\begin{align} \label{EquInterfKtier}
I\stackrel{\triangle}{=} \sum\limits_{k=1}^{K}\sum_{X_{i,k} \in \Phi_k \setminus \{ X_{1,k^\star}\}} \frac{P_k}{P_{k^\star}} \rho_{i,k} S_{i,k}|X_{i,k}|^{-\alpha_k},
\end{align}
where $\rho_{i,k}= 2^{\frac{B_{i,k}}{N_k-1}}$ for $\mathrm{BS}_{i,k} \in \mathcal{B}_L$ and $\rho_{i,k}=1$ otherwise.

\textbf{Upper Bound.} Similar to the proof for the $1$-tier case, the upper bound of SIR CDF is 
\begin{align}
\mathbb{P} \left(\gamma=\frac{S_{1,k^\star}|X_{1,k^\star}|^{-\alpha_{k^\star}}}{I_\mathcal{B}} \leq \beta\right)  
\leq \frac{\beta}{N-L-1} \left[\sum\limits_{i=2}^{\infty} \rho_{i,k} \mathbb{E}\left[\frac{|X_{1,k^\star}|^{\alpha_{k^\star}}}{|X_{i,k^\star}|^{\alpha_k^\star}}\right]+ \sum\limits_{\begin{subarray}{c}
k=1\\
k\neq k^\star
\end{subarray}}^{K} \sum\limits_{i=1}^{\infty}\frac{P_k}{P_{k^\star}} \rho_{i,k} \mathbb{E}\left[\frac{|X_{1,k^\star}|^{\alpha_{k^\star}}}{|X_{i,k}|^{\alpha_k}}\right]\right].
\end{align}
Following the same steps in Appendix \ref{ProofThmSIR1tierUB} we have
\begin{align}
\mathbb{E}\left[\frac{|X_{1,k^\star}|^{\alpha_{k^\star}}}{|X_{i,k^\star}|^{\alpha_{k^\star}}}\right]=\frac{\Gamma\left(1+ \frac{\alpha_{k^\star}}{2} \right)(i-1)!}{\Gamma\left(i+\frac{\alpha_{k^\star}}{2}\right)}
\end{align} 
For $k\neq k^\star$, $X_{i,k}$ and $X_{1,k^\star}$ belong to two independent PPPs. Therefore we have
\begin{align}\label{EquDifferentK}
\mathbb{E}\left[\frac{|X_{1,k^\star}|^{\alpha_{k^\star}}}{|X_{i,k}|^{\alpha_{k}}}\right] =\frac{\mathbb{E}\left[|X_{1,k^\star}|^{\alpha_{k^\star}}\right]}{\mathbb{E}\left[|X_{i,k}|^{\alpha_{k}}\right]}= \frac{(\lambda_{k^\star} \pi)^{-\frac{\alpha_{k^\star}}{2}}\Gamma\left(1+ \frac{\alpha_{k^\star}}{2} \right)(i-1)!}{(\lambda_k \pi)^{\frac{-\alpha_k}{2}}\Gamma\left(i+\frac{\alpha_k}{2}\right)}
\end{align}
Therefore the upper bound in Theorem \ref{ThmSIRKtier} is proven.

\textbf{Lower Bound.} The key idea of this proof is converting the interference from $K$ tiers to that of a single tier. Then we apply the lower bound from Theorem \ref{ThmSIR1tier}. Comparing (\ref{EquInterfKtier}) with (\ref{EquInterf1tier}), it is seen that the $K$-tier case is different from the $1$-tier case in two important aspects: 1) BSs from different tiers have different powers; and 2) BSs from different tiers have different path loss exponents. 

We first present the way of eliminating the power differences. The interference from the $k^{th}$ tier is
\begin{align}
I_k = \sum\limits_{X_{i,k} \in \Phi_k} P_k \rho_{i,k} S_{i,k} |X_{i,k}|^{-\alpha_k} =\sum\limits_{X_{i,k} \in \Phi_k} P_k^\star \rho_{i,k} S_{i,k} \left|\frac{X_{i,k}}{\left(P_k/P_{k^\star}\right)^{\frac{1}{\alpha_k}}}\right|^{-\alpha_k} = \sum\limits_{Y_{i,k} \in \hat{\Phi}_k} P_k^\star \rho_{i,k} S_{i,k} \left|Y_{i,k}\right|^{-\alpha_k}
\end{align}
where $Y_{i,k} \stackrel{\triangle}{=} \frac{X_{i,k}}{\left(P_k/P_{k^\star}\right)^{\frac{1}{\alpha_k}}}$. The conservation property in \cite{DSWKJM96,LiuAnd11} states that $\{Y_{i,k}, i \in \mathbb{N}\}$ form a new PPP $\hat{\Phi}_k$ with intensity $\hat{\lambda}_k=\lambda_k \left(P_k/P_{k^\star}\right)^{\frac{2}{\alpha_k}}$. Therefore the interference can be viewed as generated from the $K$ new tiers $\{\hat{\Phi}_1, \ldots,\hat{\Phi}_K\}$ with the same transmitting power $P_{k^\star}$. Therefore, the normalized interference in (\ref{EquInterfKtier}) can be rewritten as
\begin{align}
I_\mathcal{B} \stackrel{d.}{=} \sum\limits_{k=1}^{K} \sum\limits_{Y_{i,k} \in \hat{\Phi}_k \setminus \{Y_{1,k^\star}\}} \rho_{i,k} S_{i,k} |Y_{i,k}|^{-\alpha_k},
\end{align}
where $\stackrel{d.}{=}$ means equivalence in distribution.

We then set all the path loss exponents to a common value $\alpha_{\max}\stackrel{\triangle}{=} \max (\alpha_1, \ldots, \alpha_K)$, which is the best case since the interference attenuates faster. The normalized interference in (\ref{EquInterfKtier}) can be lower bounded as
\begin{align}
I_\mathcal{B} \geq \sum\limits_{k=1}^{K} \sum\limits_{Y_{i,k} \in \hat{\Phi}_k \setminus \{Y_{1,k^\star}\}} \rho_{i,k} S_{i,k} |Y_{i,k}|^{-\alpha_{\max}} \stackrel{\triangle}{=} I_{\mathcal{B}}^{lb}.
\end{align}

Denote $\hat{\Phi}= \bigcup\limits_{k=1}^K \hat{\Phi}_k$. Because $\{\hat{\Phi}_1,\ldots,\hat{\Phi}_K\}$ are independent PPPs, $\hat{\Phi}$ is also a PPP with intensity $\hat{\lambda}= \sum\limits_{k=1}^{K} \hat{\lambda}_k$. The interference lower bound $I^{lb}_{\mathcal{B}}$ can be viewed as generated from a single tier where 1) the BS locations forms a PPP $\hat{\Phi}$ with intensity $\hat{\lambda}$; 2) BSs have the same transmitting power $P_{k^\star}$; and 3) BSs have the same path loss exponent $\alpha_{\max}$. Therefore, $I_{\mathcal{B}}^{lb}$ can be rewritten as
\begin{align}
I_{\mathcal{B}}^{lb} &=\sum\limits_{X_{i,k} \in \hat{\Phi} \setminus \{X_{1,k^\star}\}}  \rho_{i,k} S_{i,k} |Y_{i,k}|^{-\alpha_{\max}} \nonumber \\
&=\sum\limits_{\mathrm{BS}_{i,k} \in \mathcal{B}} \rho_{i,k} S_{i,k} |Y_{i,k}|^{-\alpha_{\max}} + \sum\limits_{\mathrm{BS}_{i,k} \in \hat{\Phi} \setminus \{\mathrm{BS}_{1,k^\star} \bigcup \mathcal{B}\}} S_{i,k} |Y_{i,k}|^{-\alpha_{\max}}.
\end{align}
It is obvious now that $I_{\mathcal{B}}^{lb}$ is in the same form as $I_{\mathcal{B}}$ in (\ref{EquInterf1tier}).

Applying (\ref{EquLB1tier}) and (\ref{EquLBsequel1tier}), we have the lower bound on the end-user SIR
\begin{align}
\mathbb{P}\left(\gamma = \frac{S_{1,k^\star} |X_{1,k^\star}|^{-\alpha_{k^\star}}}{I_\mathcal{B}} \leq \beta \right) & \geq  \mathbb{P} \left(\frac{S_{1,k^\star} |X_{1,k^\star}|^{-\alpha_{k^\star}}}{I_{\mathcal{B}}^{lb}} \leq \beta\right) \nonumber \\
&\geq  1-\exp\left[-\pi\hat{\lambda}\Gamma\left(1+\frac{2}{\alpha_{\max}}\right)(\mathbb{E}[\hat{Z}])^{-\frac{2}{\alpha_{\max}}}\right], 
\end{align}
where $\hat{Z} = \frac{S_{1,k^\star}|X_{1,k^\star}|^{-\alpha_{k^\star}}}{\beta[3^{-\alpha_{\max}}\rho_{\min}+(2l+3)^{-\alpha_{\max}}(1-\rho_{\min})]}$, $l=|\mathcal{B}|$, and $\rho_{\min}=\min_{\mathrm{BS}_{i,k} \in \mathcal{B}} \{\rho_{i,k}\}$. From Lemma 1 and 3 in \cite{JoSanXiaAnd11}, we have $\mathbb{E}[|X_{1,k^\star}|^{-\alpha_{k^\star}}]=(\pi \hat{\lambda})^{\frac{\alpha_{k^\star}}{2}} \Gamma(1-\alpha_{k^\star}/2)$, which gives 
\begin{align}
\mathbb{E} [\hat{Z}] = \frac{(N_{k^\star}-L) (\pi \hat{\lambda})^{\frac{\alpha_{k^\star}}{2}} \Gamma \left(1- \frac{\alpha_{k^\star}}{2}\right)}{\beta[3^{-\alpha_{\max}}\rho_{\min}+(2l+3)^{-\alpha_{\max}}(1-\rho_{\min})]}.
\end{align}
Therefore the lower bound in Theorem \ref{ThmSIRKtier} follows.

\end{proof}

The bounds in Theorem \ref{ThmSIR1tier} and \ref{ThmSIRKtier} are insightful as they show clear dependence on important parameters such as the number of BS antennas $N$, the path loss exponent $\alpha$ and the number of coordinated neighbouring BSs $L$. On the other hand, the bounds in $1$-tier network case are independent of the BS spatial density $\lambda$. This is often called \textit{scale-invariance} which is a known property of interference-limited cellular networks\cite{AndBacGan10,DhiGanBacAnd11,JoSanXiaAnd11}. These bounds on the end-user SIR $\gamma_{\mathcal{B}}$ can be used with the throughput evaluation framework in Theorem \ref{ThmFramework} to derive the bounds on the end-user throughput.

\section{Numerical Results and Discussion}
In this section, we simulate the CoMP ZFBF performance under both realistic and idealized overhead messaging. We consider a 3-tier heterogeneous cellular network comprising macro (tier 1), pico (tier 2) and femto (tier 3) BSs. We simulate the scenario that the serving BS is a macrocell BS, i.e. $k^\star=1$. Notation and system parameters are given in Table \ref{table1}.
\begin{table}[h]
\caption{notation $\&$ Simulation Summary} \centering
\begin{tabular}{c|c|c}
\hline \hline notation & Description & Simulation Value\\
\hline
\hline $K$ & The number of tiers in the HCN & $3$\\
\hline $k$ & The tier index & k=$1$ (macro), $2$ (pico), $3$ (femto)\\
\hline $k^\star$ & The tier index of the serving BS & $1$ (macro) \\
\hline $P_k$  & Transmitting powers of $k^{th}$-th tier & $P_1=40 W$, $P_2=2 W$, $P_3=0.2 W$ \\
\hline $\lambda_{k}$  & Spatial density of $k^{th}$ tier & $\lambda_1=10^{-6}/m^2$, $\lambda_2=10^{-5}/m^2$, $\lambda_3=10^{-4}/m^2$ \\
\hline $\alpha_k$  & Path loss exponent of $k^{th}$ tier & $\alpha_1=4.0$, $\alpha_{2}=3.5$, $\alpha_{3}=3.0$\\
\hline $N_k$ & The number of BS antennas in the $k$-th tier & $N_1=8$, $N_2=4$, $N_3=2$\\
\hline $\mathcal{D}_{i,k}$ & Overhead delay & Not fixed\\
\hline $\tau_{i,k}$ & Time fraction of $\mathrm{BS}_{i,k}$ in cooperation phase & Not fixed \\ 
\hline $\mathcal{T}_{i,k}$ & Channel fading block length of $\mathrm{BS}_{i,k}$ & Fixed length of $80$ ms\\
\hline $B_{i,k}$ & Overhead message quantization bits for $\mathrm{BS}_{i,k}$ & Not fixed \\
\hline $L$ & The number of coordinated cells & Not fixed\\
\hline $\mathcal{B}_L$ & The set of coordinated cells & Not fixed \\
\hline $\gamma_{\mbox{target}}$ & The target SIR used in (\ref{EquCoverage}) & $3$ dB \\
\hline $G$ & The Shannon gap used in (\ref{EquThroughput}) & $3$ dB \\
\hline\hline
\end{tabular} \label{table1}
\end{table}

Two main performance metrics for CoMP schemes (including CoMP ZFBF) are their improvements on network coverage and capacity. We simulate both performance metrics, by considering the different types of SIR-rate mapping functions $R(\cdot)$ in Theorem \ref{ThmFramework}.
\begin{enumerate}
\item (Coverage under CoMP ZFBF) If the end-user only requires a fixed target rate $\mathcal{R}_{\mbox{target}}$ (e.g. a voice user), its SIR-rate mapping function $R(\cdot)$ is
\begin{align} \label{EquCoverage}
R(\gamma) =
\begin{cases}
\mathcal{R}_{\mbox{target}} & \gamma \geq \gamma_{\mbox{target}} \\
0 & \gamma < \gamma_{\mbox{target}}\\
\end{cases}.
\end{align}
The throughput quantified in Theorem \ref{ThmFramework} is then simply the user's target rate times its probability of being in coverage (i.e. with SIR $\gamma$ larger than the target SIR $\gamma_{\mbox{target}}$). Therefore, we can quantify the coverage improvement from CoMP ZFBF by normalizing the derived throughput by $\mathcal{R}_{\mbox{target}}$. We use $\gamma_{\mbox{target}}=3$ dB in the simulations.
\item (Throughput under CoMP beamforming) If the end-user is a data-greedy user, its SIR-rate mapping function $R(\cdot)$ is 
\begin{align} \label{EquThroughput}
R(\gamma)= \log_2\left(1+ \frac{\gamma}{G}\right) \quad \mbox{bps/Hz},
\end{align}
where $G$ is the Shannon gap, which is $3$ dB in our simulations. Quantifying the rate of users of this kind will show the throughput improvement from CoMP ZFBF.
\end{enumerate}

\subsection{The Configurations of Overhead Channel}
Regarding inter-cell overhead channels, we consider the scenario that the overhead messages are shared through the BSs' backhaul. In our previous work \cite{XiaJoAnd11}, the backhaul connection between two coordinated BSs is modelled as a tandem queue network consisting of several servers (e.g. switches, routers and gateways), each of which has exponential processing time. See \cite{XiaJoAnd11} for more details on this model. The limited processing rate from the backhaul servers inevitably introduces overhead delay, whose distribution $p(\mathcal{T}_{i,k},d)$ is derived in Theorem 1 in \cite{XiaJoAnd11} and used in the simulations. 

\begin{figure}[hp]
\centerline{
\includegraphics[width=3.5in]{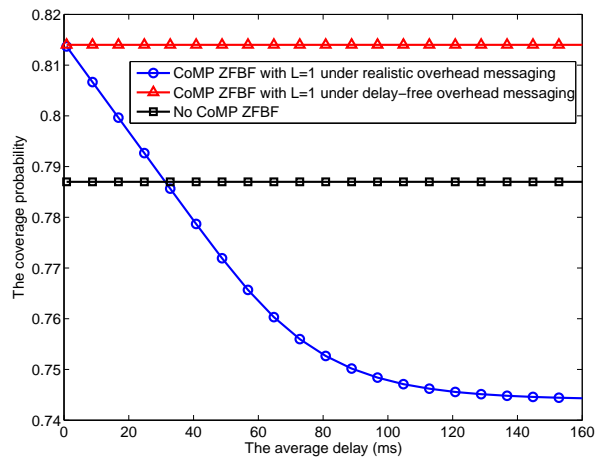}}
\caption{Downlink CoMP ZFBF coverage probability vs. the average overhead channel delay. The overhead bit size is $B_{i,k}=3(N_k-1)$, which gives $\rho_{i,k}=12.5\%$ (i.e. a coordinated $\mathrm{BS}_{i,k}$ can cancel $87.5\%$ of its interference once receiving the updated overhead). The number of coordinated cells is $L=1$.} \label{PicDelayVsCoverage}
\end{figure}

\begin{figure}[hp]
\centerline{
\includegraphics[width=3.5in]{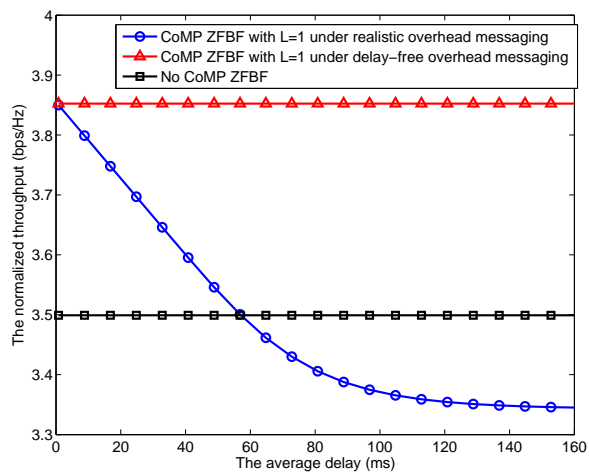}}
\caption{Downlink CoMP ZFBF throughput vs. the average overhead channel delay. The values of $L$ and $B_{i,k}$ are the same as Fig. \ref{PicDelayVsCoverage} (i.e. $B_{i,k}=3(N_k-1)$ and $L=1$).} \label{PicDelayVsRate}
\end{figure}

Fig. \ref{PicDelayVsCoverage} and \ref{PicDelayVsRate} show the impact of backhaul overhead delay on the CoMP ZFBF coverage and throughput. It is seen that CoMP ZFBF coverage and throughput fall almost linearly as the average delay grows from zero (i.e. delay-free overhead channel as assumed in previous literature). In fact, when the overhead channel delay is larger than $60\%$ of the fading coherence time $\mathcal{T}_{i,k}$, CoMP ZFBF does not bring any coverage or throughput gain. This observation diverges significantly from previous optimistic performance predictions which ignore the overhead delay. It also provides a rule of thumb for the overhead channel configurations for CoMP ZFBF.  

\subsection{Choosing Coordinated Cells}
\begin{figure}[h]
\centerline{
\includegraphics[width=3.5in]{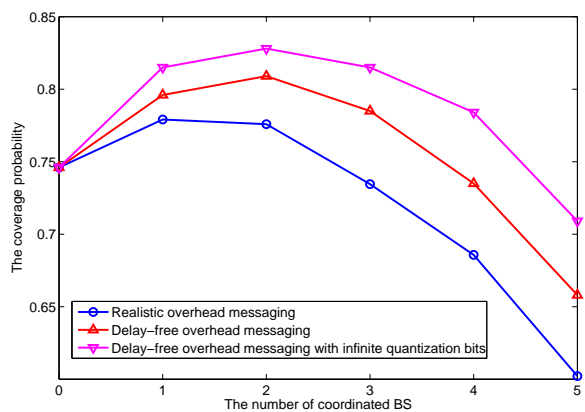}}
\caption{CoMP ZFBF coverage probability vs. the number of coordinated cells $L$. We use $B_{i,k}=3(N_k-1)$ to give $\rho_{i,k}=12.5\%$, i.e. a coordinated BS $\mathrm{BS}_{i,k}$ can cancel $87.5\%$ of its interference once receiving the overhead. For the overhead delay $\mathcal{D}_{i,k}$ between the serving cell and $\mathrm{BS}_{i,k}$, we adjust the servers' processing rates in their backhaul path, to make sure that the average overhead delay $\mathbb{E}[\mathcal{D}_{i,k}]=20$ ms\cite{HuaWei09}. } \label{PicLvsCoverage}
\end{figure}
\begin{figure}[h]
\centerline{
\includegraphics[width=3.5in]{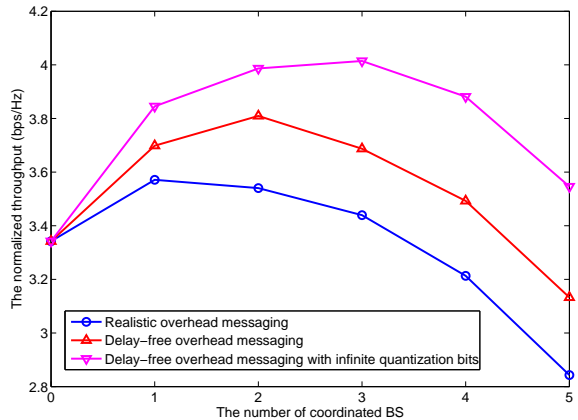}}
\caption{CoMP ZFBF throughput (bps/Hz) vs. the number of coordinated cells $L$. The configurations on $B_{i,k}$ and the backhaul channel are the same as Fig. \ref{PicLvsCoverage} (i.e. $B_{i,k}=3(N_k-1)$ and $\mathbb{E}[\mathcal{D}_{i,k}]=20$ ms). } \label{PicLvsRate}
\end{figure}
A fundamental design question for CoMP beamforming is how many and which neighbouring cells should be coordinated. Coordinating more cells may translates into less interference from other cells, but also weaker signal power for the end-user and heavier overhead messaging burden. Fig. \ref{PicLvsCoverage} and \ref{PicLvsRate} show that the number of coordinated cells should be kept fairly small in HCNs, even under ideal overhead model (i.e. no overhead delay and infinite quantization bits) and limited feedback overhead model (i.e. no overhead delay but finite quantization bits). This observation diverges significantly from previous work in conventional macrocell networks and implies that dominant interference comes from only a few neighboring cells in HCNs. Further, by considering both overhead delay and rate constraints, our work indicates that coordinating with only one other cell (i.e $L=1$) is actually optimal for a serving base station $\mathrm{BS}_{i,k^\star}$ with eight antennas. 

\begin{figure}[h]
\centerline{
\includegraphics[width=3.5in]{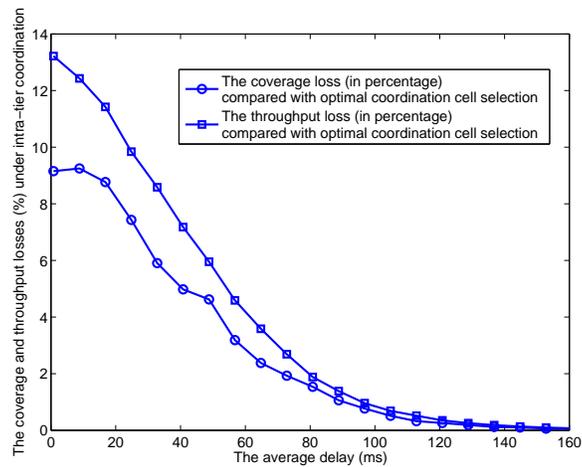}}
\caption{CoMP ZFBF coverage and throughput losses (in percentage) under intra-tier coordination vs. the average overhead channel delay. The values of $L$ and $B_{i,k}$ are the same as Fig. \ref{PicDelayVsCoverage} (i.e. $B_{i,k}=3(N_k-1)$ and $L=1$).} \label{PicDelayvsCoverageAndRate}
\end{figure}

When the serving base station $\mathrm{BS}_{1,k^\star}$ coordinates with $L$ other cells, we assume the optimal scenario that the $L$ strongest interfering cells (possibly coming from different tiers) are able to coordinate with the serving cell. In practice, cross-tier coordination may be restricted in HCNs and only intra-tier coordination is allowed. For example, the end-user installed femtocells are controlled by their owners, and may not be allowed to or capable of coordinating with macrocells\cite{And12JSAC}. Coordinating only with cells in the same tier (termed \emph{intra-tier coordination}) is of course sub-optimal and results in coverage and throughput losses, which are quantified in Fig. \ref{PicDelayvsCoverageAndRate} (w.r.t. different overhead delay profile) and \ref{PicLvsCoverageAndRate} (w.r.t. different values of $L$) for CoMP ZFBF. The performance loss is most significant (e.g. as high as $20\%$) in the practically important scenario of CoMP ZFBF -- that is, relatively large number of coordinated cells and/or small overhead channel delay (e.g. if the overhead channel is optimized as in \cite{Irmeretal11CoMPTrial}).
\begin{figure}[h]
\centerline{
\includegraphics[width=3.5in]{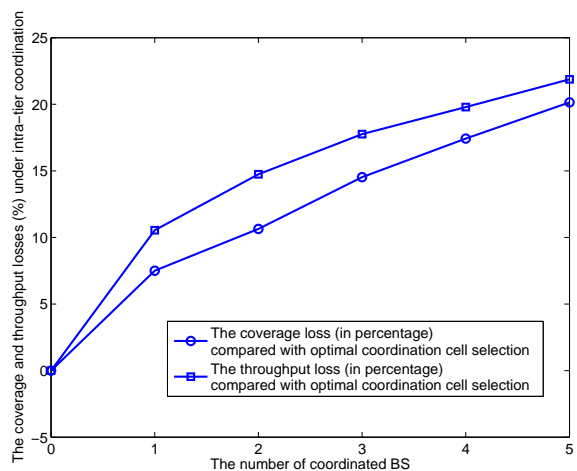}}
\caption{CoMP ZFBF coverage and throughput losses (in percentage) under intra-tier coordination vs. the number of coordinated cells $L$. The configurations on $B_{i,k}$ and the backhaul channel are the same as Fig. \ref{PicLvsCoverage} (i.e. $B_{i,k}=3(N_k-1)$ and $\mathbb{E}[\mathcal{D}_{i,k}]=20$ ms).} \label{PicLvsCoverageAndRate}
\end{figure}

\section{Conclusion}
This paper presents a novel approach to evaluate downlink CoMP schemes in the new network paradigm of HCNs, by developing a new framework to quantify the impact of overhead delay and using PPP model in end-user SINR characterization. This proposed approach can be used for a class of CoMP schemes, and is applied to CoMP ZFBF in this paper as an example. We show that CoMP ZFBF performance heavily depends on the overhead delay and its design should be fairly conservative (e.g. coordinating with only one or two other cells). These results align with the findings from several industrial implementations\cite{Qcomm10CTW,Irmeretal11CoMPTrial,Qcom12ITA}, and provide insights diverging significantly from previous work without overhead delay modeling.

In this paper, one fundamental assumption is the i.i.d. discrete time model of the cooperation-dependent parameters. Developing more complicated and accurate models is an important related topic for future works. Besides, the evaluation framework assumes the number of coordinated cells $L$ as a pre-determined parameter. The adaptive selection on $L$ as a function of instantaneous other-cell interference can potentially bring more CoMP gains, and should be considered in future CoMP evaluation in HCNs. Considering multiple user antennas in CoMP is also a related topic for future research.

\appendix
\subsection{Proof of Lemma \ref{LemmaTimeFraction}} \label{ProofLemmaTimeFraction}
The cooperation phase only occurs in fading blocks for which the overhead messaging delay $\mathcal{D}$ is smaller than $\mathcal{T}$. Here we omit the subscripts of delay $\mathcal{D}$ and block length $\mathcal{T}$ to keep the proof general. The percentage of these fading blocks is 
\begin{align}
\mathbb{P} (\mathcal{D} \leq \mathcal{T}) = \mathbb{P} (\mathcal{D}\leq \mathcal{T},\mathcal{D}\leq \infty) = p(\mathcal{T},\infty),
\end{align}
where the last equality holds by definition. In these fading blocks, the overhead messaging phase will have a time length $\mathbb{E}[\mathcal{D} | \mathcal{D}<\mathcal{T}]$
\begin{align}
\mathbb{E} [\mathcal{D} | \mathcal{D}<\mathcal{T}] = \int\limits_{0}^{\infty} \{ 1- \mathbb{P} (\mathcal{D} \leq s | \mathcal{D} \leq \mathcal{T})\} \mathrm{d}s.
\end{align}
Thus the average duration of the cooperation phase is $\mathbb{E} [\mathcal{T}] - \mathbb{E} [\mathcal{D}|\mathcal{D} \leq \mathcal{T}]$.

In sum, the long-term time fraction of the cooperation phase is 
\begin{align}
\eta & = p(\mathcal{T},\infty)\times \left(\frac{\mathbb{E} [\mathcal{T}] - \mathbb{E} [\mathcal{D}|\mathcal{D} \leq \mathcal{T}]}{\mathbb{E} [\mathcal{T}] }\right) \nonumber\\
&=p(\mathcal{T},\infty) - \frac{p(\mathcal{T},\infty)}{\mathbb{E} [\mathcal{T}]} \int\limits_{0}^{\infty} \{ 1- \mathbb{P} (\mathcal{D} \leq s | \mathcal{D} \leq \mathcal{T})\} \mathrm{d}s \nonumber \\
&\stackrel{(a)}{=} p(\mathcal{T},\infty)-\frac{1}{\mathbb{E}[\mathcal{T}]} \int\limits_{0}^{\infty} \left\{p(\mathcal{T},\infty)-p(\mathcal{T},s)\right\} \mathrm{d}s.
\end{align}
By definition, $p(\mathcal{T},\infty)=\mathbb{P} (\mathcal{D}\leq \mathcal{T},\mathcal{D} \leq \infty)=\mathbb{P} (\mathcal{D} \leq \mathcal{T})$. Therefore (a) follows.

\subsection{Auxiliary Result for the CDF Upper Bound in Theorem \ref{ThmSIR1tier} } \label{ProofThmSIR1tierUB}
Consider a one-dimensional PPP $\Phi =\{Y_1, Y_2, \ldots\}$ with intensity $\lambda$, we have 
\begin{align}
Y_i =\sum\limits_{j=1}^{i} \Delta_j
\end{align}
where $\Delta_j \sim \exp(\lambda)$. For an arbitrary positive number $\nu$, the following equality holds
\begin{align} 
\mathbb{E} \left[\frac{Y_1^\nu}{Y_i^{\nu}}\right] &=\mathbb{E} \left[\left(\frac{\Delta_1}{\Delta_1 + \Delta_2+\ldots\Delta_i}\right)^{\nu}\right]
=\mathbb{E}\left\{\mathbb{E} \left[\left(\frac{\Delta_1}{\Delta_1 + \Delta_2+\ldots\Delta_i}\right)^{\nu}\bigg{|}\Delta_1+\ldots\Delta_i=x\right]\right\} \nonumber \\
& =\mathbb{E} \left\{\frac{\Gamma\left(1+\nu\right)(i-1)!}{\Gamma\left(i+\nu\right)}\right\} =\frac{\Gamma\left(1+\nu\right)(i-1)!}{\Gamma\left(i+\nu\right)}
\end{align} 

\subsection{Auxiliary Result for the CDF Lower Bound in Theorem \ref{ThmSIR1tier}} \label{ProofThmSIR1tierLB}
For an arbitrary $m$, $I_{(m)}$ can be expressed as
\begin{align}
I_{(m)} & = \sum_{X_i \in \Phi \backslash\{X_1,\ldots,X_m\}} S_i |X_i|^{-\alpha}\nonumber \\
& = \sum_{X_i \in \Phi \backslash\{X_1,\ldots,X_m\}} S_i \left(|X_i|-|X_m|+|X_m|\right)^{-\alpha} \nonumber \\
&=\sum_{X_i \in \Phi \backslash\{X_1,\ldots,X_m\}} \left(1+\frac{|X_m|}{|X_i|-|X_m|}\right)^{-\alpha} S_i (|X_i|-|X_m|)^{-\alpha} \nonumber \\
&\geq \sum_{X_i \in \Phi \backslash\{X_1,\ldots,X_m\}} \left(1+\frac{|X_{m}|}{|X_{m+1}|-|X_m|}\right)^{-\alpha} S_i (|X_i|-|X_m|)^{-\alpha} \nonumber \\
&=\left(1+\frac{|X_{m}|}{|X_{m+1}|-|X_m|}\right)^{-\alpha} \sum_{X_i \in \Phi \backslash\{X_1,\ldots,X_m\}}  S_i (|X_i|-|X_m|)^{-\alpha}
\end{align}
Since the above inequality holds for any realization of $|X_{m}|$ and $|X_{m+1}|$, we have
\begin{align} \label{EquImLB}
I_{(m)} &\stackrel{a.s.}{\geq} \left(1+\frac{\mathbb{E}[|X_{m}|]}{\mathbb{E}[|X_{m+1}|]-\mathbb{E}[|X_m|]}\right)^{-\alpha} \sum_{X_i \in \Phi \backslash\{X_1,\ldots,X_m\}}  S_i (|X_i|-|X_m|)^{-\alpha} \nonumber\\
&\stackrel{(a)}{=} (1+2m)^{-\alpha} \sum_{X_i \in \Phi \backslash\{X_1,\ldots,X_m\}}  S_i (|X_i|-|X_m|)^{-\alpha},
\end{align}
where (a) follows because we have $\mathbb{E} [|X_m|]= (\lambda \pi)^{-0.5} \frac{\Gamma(m+0.5)}{(m-1)!}$.

Now we define another point process $\tilde{\Phi}_m\stackrel{\triangle}{=}\{Y_j\in \mathbb{R}^2: \mbox{ for any, } Y_{j}=X_{i}-\frac{X_i}{|X_i|}|X_m|, j=i-m\}$, i.e. $\tilde{\Phi}$ is formed by moving every point of $\{X_{m+1}, X_{m+2},\ldots\}$ in the PPP $\Phi$ toward the origin by distance $|X_m|$. Note that $\tilde{\Phi}_m$ is also a spatial Poisson Point Process, but non-homogeneous with a larger density than the original $\Phi$ (because when moving points toward the origin, we actually compress the space in $\mathbb{R}^2$). Therefore the sum interference from $\tilde{\Phi}_m$ is larger than that from the original $\Phi$, i.e. larger than $I_{(0)}$. 
\begin{align}
I_{(m)} &\stackrel{a.s. }{\geq} (1+2m)^{-\alpha} \sum\limits_{i=m+1}^{\infty}  S_i (|X_i|-|X_m|)^{-\alpha} \nonumber \\
&=(1+2m)^{-\alpha} \sum\limits_{Y_j \in \tilde{\Phi}_m}^{\infty} S_{j+m} |Y_j|^{-\alpha} \nonumber \\
&  \succeq (1+2m)^{-\alpha} I_{(0)}.
\end{align}

\bibliographystyle{IEEEtran}
\bibliography{DownlinkZFBF}

\end{document}